\documentclass[10pt,twoside]{IEEEtran} 
\usepackage{setspace}

\ifCLASSOPTIONcompsoc
\usepackage[nocompress,noadjust]{cite}
\else
\usepackage[noadjust]{cite}
\fi
\usepackage{graphicx} 
\usepackage{mathptmx} 
\usepackage{times} 
\usepackage{amsmath} 
\usepackage{amssymb}  
\usepackage{amsthm}

\usepackage{tikz}
\usepackage{subfigure}

\theoremstyle{definition}
\newtheorem{definition}{Definition}[section]

\theoremstyle{plain}
\newtheorem{assumption}[definition]{Assumption}
\newtheorem{lemma}[definition]{Lemma}
\newtheorem{proposition}[definition]{Proposition}
\newtheorem{theorem}[definition]{Theorem}
\newtheorem{corollary}[definition]{Corollary}

\theoremstyle{remark}
\newtheorem{remark}[definition]{Remark}

\newcommand{\abs}[1]{\lvert #1 \rvert}

\newcommand\ssp[1]{\mathchoice{#1^\prime}{#1^\prime}{#1^\prime}%
{#1^{\scalebox{.7}{$\scriptscriptstyle\prime$}}}
}

\DeclareMathOperator{\col}{col}

\DeclareSymbolFont{bbold}{U}{bbold}{m}{n}
\DeclareSymbolFontAlphabet{\mathbbold}{bbold}
\newcommand{\onev}{\mathbbold1}

\newcommand{\myindent}{\hspace{.2cm}}

\usepackage{calrsfs}
\DeclareMathAlphabet{\pazocal}{OMS}{zplm}{m}{n}
\renewcommand{\mathcal}[1]{\pazocal{#1}}

\usepackage{mathtools}

\allowdisplaybreaks

\begin{document}
%
\title{Stability of SIS Spreading Processes in Networks with Non-Markovian Transmission and Recovery}

\author{Masaki~Ogura,~\IEEEmembership{Member,~IEEE,}
and~Victor~M.~Preciado,~\IEEEmembership{Member,~IEEE}
\IEEEcompsocitemizethanks{\IEEEcompsocthanksitem M. Ogura is with the Division 
of Information Science, Nara Institute of Science and Technology, Ikoma, Nara 630-0192, Japan.
E-mail: \mbox{oguram@is.naist.jp}
\IEEEcompsocthanksitem V. M. Preciado is with the Department of Electrical and Systems Engineering, University of Pennsylvania, Philadelphia, PA 19104, USA.
\mbox{E-mail: preciado@seas.upenn.edu}
}
}

%

\IEEEtitleabstractindextext{%
\begin{abstract}
Although viral spreading processes taking place in networks are {often} analyzed using Markovian models in which both the \emph{transmission} and the \emph{recovery} times follow exponential distributions, empirical studies show that, in {many} real scenarios, the distribution of these times are {not necessarily} exponential. To overcome this limitation, we first introduce a generalized {susceptible-infected-susceptible (SIS)} spreading model that allows transmission and recovery times to follow \emph{phase-type} distributions. In this context, we derive a lower bound on the exponential decay rate towards the infection-free equilibrium of the spreading model without relying on mean-field approximations. Based on our results, we illustrate how the particular shape of the transmission/recovery distribution influences the exponential rate of convergence towards the equilibrium. 
\end{abstract}
}

\maketitle

\IEEEdisplaynontitleabstractindextext

%
\IEEEpeerreviewmaketitle

\section{Introduction}

\IEEEPARstart{U}{nderstanding} the dynamics of spreading processes in complex networks is a challenging problem with a wide range of practical applications in epidemiology and public health~\cite{Bailey1975}, information propagation in social networks~\cite{Lerman2010}, or cyber-security~\cite{Roy2012}. During the last decade, significant progress has been made towards understanding the relationship between the topology of a network and the dynamics of spreading processes taking place over the network (see~\cite{Pastor-Satorras2015a, Nowzari2015a} for recent surveys). A common approach to investigate this relationship is by modeling spreading processes using networked Markov processes, such as the networked susceptible-infected-susceptible (SIS) model~\cite{VanMieghem2009a}. Based on these Markovian models, it is then possible to find an explicit relationship between epidemic thresholds and network eigenvalues in static topologies {\cite{Wan2008IET,VanMieghem2009a,Khanafer2016,Pare2018a}}, as well as in multilayer {\cite{DarabiSahneh2013,Santos2015}}, time-varying {\cite{Ogura2015c,Pare2018}}, and adaptive {\cite{Ogura2015i,Mai2018}} networks. Markovian models also allow us to design optimal strategies for containment of spreading processes taking place in static~\cite{Preciado2014}, uncertain~\cite{Han2015a}, and temporal~\cite{Ogura2016l} networks.

A consequence of using Markovian models in the analysis of spreading processes is that both \emph{transmission times} (i.e., the time it takes for an infection to be transmitted from an infected node to one of its neighbors) and \emph{recovery times} (i.e., the time it takes for an infected node to recover) follow exponential distributions. However, empirical studies show that, in {many} real networks, the distribution of transmission and recovery times are {not necessarily} exponential. For example, the transmission time of messages in Twitter, or news in other social media outlets, follows (approximately) a log-normal distribution~\cite{Lerman2010,VanMieghem2013}. In the context of human contact networks, the transmission of various infectious diseases~\cite{Blythe1988,Chowell2014,Keeling2002,NISHIURA2007}, or the time it takes to recover from the influenza virus~\cite{Suess2010} are {often} non-exponential.

Since realistic transmission and recovery times {often} follow non-exponential distributions, it is of practical importance to understand the role of these distributions on the dynamics of the spread. In this direction, the authors in~\cite{VanMieghem2013} illustrated, via numerical simulations, that non-exponential transmission times can have a substantial effect on the dynamics of the spread. Motivated by this study, several approximative methods for quantifying the steady-state fraction of infected nodes have been proposed in the literature. In this direction, the authors in~\cite{Cator2013a} analyzed spreading processes with general transmission and recovery times using mean-field approximations. In~\cite{Min2013,Jo2014}, {simple but yet} analytically solvable spreading models with non-exponential transmission times were studied. Moment\nobreakdash-closure approximations for analyzing spreading processes with non-exponential transmission and recovery times were proposed in~\cite{Kiss2015,Pellis2015}. Under the assumption that recovery times follow an exponential distribution, the analytical framework in~\cite{Starnini2017} enables us to reduce non-exponentially distributed transmission times into exponentially distributed counterparts without changing the steady-state of the spread. {The authors in \cite{Nowzari2015} used a mean-field approximation to derive stability conditions for the infection-free equilibrium of a spreading process with three compartments and non-Markovian transition dynamics. However, their conditions are either conservative or guarantees only the local stability of the infection-free equilibrium.}

In this paper, we propose a tractable but rigorous approach to analyze the \emph{transient} of {SIS} spreading processes over arbitrary networks with general (non-exponential) transmission and recovery times. In this direction, we first introduce the \emph{generalized networked SIS} (GeNeSIS) model, which allows for transmission and recovery times following arbitrary phase\nobreakdash-type distributions (see, e.g.,~\cite{Asmussen1996}). {Defined as the \emph{exit time} of time-homogeneous Markov processes,} phase-type distributions form a dense family in the space of positive-valued distributions~\cite{Cox1955}. Therefore, the GeNeSIS model allows to theoretically analyze arbitrary transmission and recovery times within an arbitrary accuracy~\cite{Asmussen1996}. We are particularly interested in quantifying the exponential decay rate of the spread towards the infection-free equilibrium; in other words, to eradicate the viral spreading process. The key tool used in our derivations is a vectorial representation of phase-type distributions, which we use to bound the exponential decay rate towards the  infection-free equilibrium in the stochastic dynamics of the GeNeSIS model. 

This paper is organized as follows. In Section~\ref{sec:math}, we introduce elements of graph theory and stochastic differential equations with Poisson jumps. In Section~\ref{sec:SISmodel}, we describe a generalized {SIS} model over networks with arbitrary transmission and recovery times. In Section~\ref{sec:vectorRep}, we provide a vectorial representation of the GeNeSIS model, which we use in Section~\ref{sec:analysis} to analyze the exponential decay rate towards the infection-free equilibrium. We validate the effectiveness of our results via numerical simulations in Section~\ref{sec:sim}, where we also illustrate the effect of non-exponential transmission/recovery times in the dynamics of the spread.

\section{Mathematical Preliminaries} \label{sec:math}

Let $\mathbb{R}$ and~$\mathbb{N}$ denote the set of real numbers and positive integers, respectively. For a positive integer $n$, define $[n] = \{1, \dotsc, n\}$. For a real function $f$, let $f(t^-)$ denote the limit of $f$ from the left at time~$t$. We let $I_n$ and~$O_n$ denote the $n\times n$ identity and zero matrices. By $\onev_p$ and~$0_p$, we denote the $p$-dimensional vectors whose entries are all ones and zeros, respectively. A real matrix~$A$ (or a vector as its special case) is said to be nonnegative, denoted by $A\geq 0$, if $A$ is nonnegative entry-wise. The notation $A\leq 0$ is understood in the obvious manner. For a square matrix $A$, the maximum real part of its eigenvalues, called the \emph{spectral abscissa} of $A$, is denoted by~$\eta(A)$. We say that $A$ is Metzler if the off-diagonal entries of $A$ are all non-negative. It is easy to see that, if $A$ is Metzler, then $e^{At}\geq 0$ for every $t \geq 0$. The Kronecker product~\cite{Horn1990} of two matrices~$A$ and~$B$ is denoted by $A\otimes B$, and the Kronecker sum of two square matrices~$A\in \mathbb{R}^{p\times p}$ and~$B\in\mathbb{R}^{q\times q}$ is defined by
\begin{equation*}
A\oplus B = A\otimes I_q + I_p\otimes B.
\end{equation*}
Given a collection of $n$ matrices~$A_1,\ldots, A_n$ having the same number of columns, the matrix obtained by stacking the matrices in vertical ($A_1$ on top) is denoted by~$\col(A_1, \dotsc, A_n)$.

An undirected graph is a pair~$\mathcal G=(\mathcal V, \mathcal E)$, where
$\mathcal V = \{1, \dotsc, n\}$ is the set of nodes, and~$\mathcal E \subset
\mathcal V \times \mathcal V$ is the set of edges, consisting of distinct and
unordered pairs~$\{i, j\}$ for $i, j\in \mathcal V$. We say that a node~$i$ is a
neighbor of~$j$ (or that $i$ and~$j$ are adjacent) if $\{i, j\} \in \mathcal E$.
The set of neighbors of node~$i$ is denoted by~$\mathcal N_i$. The adjacency
matrix of $\mathcal G$ is defined as the $n\times n$
matrix whose $(i,j)$-th entry is $1$ if and only if nodes
$i$ and~$j$ are adjacent, $0$ otherwise.

We let $P(\cdot)$ denote the probability of events. The expectation of a random variable is denoted by~$E[\cdot]$. A Poisson counter~\cite[Chapter~4]{Cinlar1975a} of rate $\lambda>0$ is denoted by $N_\lambda$. In this paper, we extensively use a specific class of stochastic differential equations with Poisson jumps, described below. For each $i\in [m]$, let $f_i\colon \mathbb{R}\times \mathbb{R} \to \mathbb{R}$ be a continuous function, $N_{\lambda_i}$ be a Poisson counter, and~$\kappa_i$ be a continuous-time, real, and stationary stochastic process defined over the probability space~$\Omega$. All the above stochastic processes are assumed to be independent of each other. Then, we say that a real and right\nobreakdash-continuous function~$x$ is a solution of the stochastic differential equation
\begin{equation} \label{eq:SIE}
dx = \sum_{i=1}^m f_i(x(t), \kappa_i(t))\,dN_{\lambda_i}, 
\end{equation}
if $x$ is constant on any interval where none of the counters~$N_{\lambda_1}$,
$\dotsc$,~$N_{\lambda_m}$ jumps, and
\begin{equation*}
x(t)
=
x(t^{-})+ f_i(x(t^{-}),\kappa_i(t))
\end{equation*}
when $N_{\lambda_i}$ jumps at time~$t$. This definition can be naturally
extended to the vector case. Below, we present two lemmas for this class of
stochastic differential equations. The first lemma states a version of It\^o's
formula:

\begin{lemma}\label{lem:Ito}
Assume that $x$ is a solution of \eqref{eq:SIE}. Let $g$ be a real continuous
function. Then, $y(t) = g(x(t))$ is a solution of the stochastic differential
equation 
\begin{equation*}
dy 
= 
\sum_{i=1}^{m}\left[ g\left(x(t) +
f_i\left(x(t),\kappa_i(t)\right)\right) - g(x(t))\right] dN_{\lambda_i}, 
\end{equation*}
i.e., $y$ satisfies
\begin{equation*}
y(t) 
= 
y(t^-) + g\left(x(t^-) + f_i(x(t^-),\kappa_i(t))\right) - g(x(t^-)),
\end{equation*}
if the Poisson counter $N_{\lambda_i}$ jumps at time~$t$ and is
constant over any interval in which none of the counters
$N_{\lambda_1},\ldots,N_{\lambda_m}$ jumps.
\end{lemma}

\begin{proof}
Assume that the counter $N_{\lambda_i}$ jumps at time~$t$. Since $x$ is the
solution of the stochastic differential equation~\eqref{eq:SIE}, it follows that
$y(t)-y(t^-) = g\left(x(t^-) + f_i(x(t^-), \kappa_i(t))\right) - g(x(t^-)),$ as
desired.
\end{proof}

We also state the following lemma concerning the expectation of the solution to
the stochastic differential equation~\eqref{eq:SIE}:

\begin{lemma}\label{lem:expectation}
Assume that $x$ is a solution of \eqref{eq:SIE}. If the functions $f_1$,
$\dotsc$, $f_m$ are affine with respect to the second variable, then
\begin{equation}\label{eq:lem:expectation}
\frac{d}{dt}E[x(t)] = \sum_{i=1}^m E\bigl[f_i(x(t), E[\kappa_i(t)])\bigr]\lambda_i.
\end{equation}
\end{lemma}

\begin{proof}
By the assumption, for each $i\in [m]$ we can take real functions~$f_{i,1}$
and~$f_{i,2}$ such that $f_i(a,b) = f_{i,1}(a)+bf_{i,2}(a)$ for every $a,
b\in\mathbb{R}$. Let $t\geq 0$ and~$h>0$ be arbitrary. We have the following
three possibilities: (\emph{i}) no counter jumps on the time interval $[t,
t+h]$; (\emph{ii}) exactly one counter jumps in the interval; or (\emph{iii})
more than one counter jumps in the interval. The first case happens with
probability $1-(\lambda_1 + \cdots + \lambda_m )h+o(h)$. For the second case,
for each $i \in [m]$, one and the only one counter~$N_{\lambda_i}$ jumps on the
time interval~$[t, t+h]$ with probability $\lambda_ih + o(h)$. In this case we
have
\begin{equation*}
\begin{aligned}
x(t+h) 
&= x(t) + f_i(x(t), \kappa_i(\tau)) 
\\
&= x(t) + f_{i,1}(x(t)) +
\kappa_i(\tau)f_{i,2}(x(t))
\end{aligned}
\end{equation*} 
for some $\tau \in [t, t+h]$ and, therefore,
\begin{equation*}
\begin{aligned}
E[x(t+h)] 
&= 
E[x(t)] + E[f_{i,1}(x(t))] + E[\kappa_i(\tau)]E[f_{i,2}(x(t))] 
\\
&=
E[x(t)] + E\bigl[f_i(x(t), E[\kappa_i(t)])\bigr]
\end{aligned}
\end{equation*}
because $\kappa_i$ is a stationary stochastic process. Finally, the third case
occurs with probability $o(h)$. Summarizing, we have shown that
\begin{equation*}
\frac{E[x(t+h)]-E[x]}{h} 
= 
\frac{o(h)}{h} + \sum_{i=1}^{m}E\bigl[f_i(x(t), E[\kappa_i (t)])\bigr]\lambda_i, 
\end{equation*}
which proves \eqref{eq:lem:expectation} in the limit of $h\to 0$.
\end{proof}

\section{SIS Model with General Transmission and Recovery Times}\label{sec:SISmodel}

The aim of this section is to introduce the generalized networked susceptible-infected-susceptible (GeNeSIS) model, which will allow us to analyze the effect of {non-exponential} transmission and recovery times in the spreading dynamics.

\subsection{Generalized Networked SIS Model}\label{sec:gensis}

{We start by giving a brief overview of the standard SIS model (see, e.g., \cite{Nowzari2015a,Pastor-Satorras2015a}). Let $\mathcal G = (\mathcal V, \mathcal E)$ be an undirected and unweighted graph with $n$ nodes. In the SIS model, at a given (continuous) time~$t \geq 0$, each node can be in one of two possible states: {\it susceptible} or {\it infected}. If a neighbor of node~$i$ is infected, then this neighbor can infect node~$i$ with an instantaneous rate~$\beta_i$, where $\beta_i > 0$ is called the \emph{transmission rate} of node $i$. Therefore, while being infected, the neighbor attempts to infect node $i$ with the inter-event times following an exponential distribution of rate~$\beta_i$. On the other hand, when a node~$i$ is infected, it can randomly transition to the susceptible state with an instantaneous rate~$\delta_i > 0$, called the \emph{recovery rate} of node~$i$. This implies that the time it takes for an infected node $i$ to recover follows an exponential distribution of rate $\delta_i$.}

{Before we introduce the generalized networked susceptible-infected-susceptible (GeNeSIS) model, we introduce the following notations.} We describe the state of a node~$i\in \mathcal V$ by a $\{0,1\}$-valued continuous-time stochastic process, denoted by $z_i =\{z_i(t)\}_{t \in \mathbb{R}}$. We say that node~$i$ is susceptible (respectively, infected) at time~$t$ if $z_i(t) = 0$ (respectively, $z_i(t) = 1$). We assume that the function~$z_i$ is continuous from the right for all $i\in [n]$. Under this assumption, we say that node~$i$ becomes infected (respectively, becomes susceptible) at time~$t$ if $z_i(t^-)=0$ and~$z_i(t) = 1$ (respectively, $z_i(t^-)=1$ and~$z_i(t) = 0$). It is assumed that all nodes are susceptible before time~$t=0$, i.e., $z_i(t) = 0$ for $t<0$. We now introduce the GeNeSIS model as follows.

\begin{definition}\label{defn:GSIS}
We say that the family $z = \{z_i\}_{i \in [n]}$ of stochastic processes is a
\emph{generalized networked susceptible-infected-susceptible model}
(\emph{GeNeSIS model, for short}) if there exist a subset $\mathcal V_0 \subset[n]$,
as well as random variables $$0 = \tau_0^{ji}(t)<\tau_1^{ji}(t)<\cdots,$$ and
$\rho^i(t)>0$ satisfying the following conditions for all $i\in[n]$,
$j\in\mathcal N_i$, and~$t\geq 0$:
\begin{enumerate}

\item[a)] Node~$i$ becomes infected at time~$t=0$ if and only if $i\in \mathcal V_0$,
i.e., $\mathcal V_0$ is the initially infected subset.

\item[b)] Assume that node~$i$ becomes infected at time~$t$. Then, node~$i$
remains infected during the time interval $[t, t+\rho^i(t))$ and becomes
susceptible at time~$t+\rho^i(t)$, i.e., the random variable $\rho^i(t)$ is the
\emph{recovery time} of node~$i$;

\item[c)] If node $i$ becomes infected at time~$t$, then, until its recovery, the node attempts to infect node~$j\in\mathcal N_i$ at times $\{t+\tau_k^{ji}(t)\}_{k\in \mathbb{N}}${, i.e., if}  node~$j$ is susceptible at time~$t+\tau_k^{ji}(t)$ for any $k\in \mathbb{N}$, then node~$j$ becomes infected.
\end{enumerate}
\end{definition}

\begin{remark}
We call the random increments $\{\tau_{k}^{ji}(t) - \tau_{k-1}^{ji}(t)\}_{j\in
\mathcal N_i,\,t\geq 0,\,k\in\mathbb{N}}$ the \emph{transmission times} of
node~$i$, since the difference~$\tau_{k}^{ji}(t) - \tau_{k-1}^{ji}(t)$
represents the time between infection attempts from an infected node~$i$ towards
a neighboring node~$j$. Note that, when all the recovery and transmission times
follow exponential distributions, the GeNeSIS model recovers the standard
networked SIS model {described at the beginning of this subsection}.
\end{remark}

Notice that the origin (i.e., $z_i= 0$ for all $i\in\mathcal V$) is an absorbing state of the GeNeSIS dynamics. In what follows, we will refer to the origin as the \emph{infection-free} equilibrium. The aim of this paper is to quantify the transient dynamics of the generalized SIS model according to the following definition:

\begin{definition}\label{defn:decay}
The \emph{exponential decay rate} of the GeNeSIS model is defined by
\begin{equation*}
\lambda = - \sup_{\mathcal V_0 \subset [n]} \limsup_{t\to\infty} \frac{\log \sum_{i=1}^n E[z_i(t)]}{t}. 
\end{equation*}
\end{definition}

{Since the sum~$\sum_{i=1}^n E[z_i(t)]$ equals the expected number of infected nodes at time $t$, the decay rate~$\lambda$ quantifies how fast the infectious spreading process dies out in the network (in average). Besides quantifying the impact of contagious spreading processes over networks~\cite{Lajmanovich1976,Ganesh2005,VanMieghem2013}, the exponential decay rate has been used as a standard tool for measuring the performance of strategies aiming to contain epidemic outbreaks~\cite{Wan2008IET,Preciado2014,Han2015a,AbadTorres2016}. We further remark that, although exponential \emph{distributions} are not necessarily appropriate for modeling realistic transmission and recovery times as discussed in the Introduction, the exponential \emph{decay rate} is still a valid quantity for measuring the spreading capability of epidemic processes. }

\subsection{Phase-type Transmission and Recovery Times}

In this paper, we consider the GeNeSIS model with transmission and recovery
times following phase-type distributions~\cite{Asmussen1996}. In what follows,
we briefly describe this class of probability distributions. Consider a
time-homogeneous Markov process~$x$ in continuous-time with $p+1$ ($p\in
\mathbb{N}$) states (also called \emph{phases}) such that the
states~$1,\dotsc,p$ are transient and the remaining state~$p+1$ is absorbing.
The infinitesimal generator of the process is then necessarily of the form
\begin{equation}\label{eq:inf-gen}
\begin{bmatrix}
T & b\\
0 & 0
\end{bmatrix}
,\quad 
b = -T\onev_p, 
\end{equation}
where $T \in \mathbb{R}^{p\times p}$ is an invertible Metzler matrix with
non-positive row-sums. Let 
\begin{equation*}
\begin{bmatrix}
\phi\\0 
\end{bmatrix}\in \mathbb{R}^{p+1}\quad (\phi \in
\mathbb{R}^p)
\end{equation*}
denote the initial distribution of the Markov process $x$, i.e., 
\begin{equation*}
P(x(0) = m) = \begin{cases}
\phi_m,&m\in[p], 
\\
0,&m=p+1. 
\end{cases}
\end{equation*}
Then, the time to absorption into the state~$p+1$  is a random variable following a \emph{phase-type distribution}, which we denote by the pair~$(\phi, T)$. In the rest of the paper, we make the following assumption on the distribution of (random) transmission and recovery times in the GeNeSIS model:

\begin{assumption}\label{assm:phaseTYpe}
Transmission and recovery times of all nodes follow phase-type
distributions $(\phi, T)$ and~$(\psi, R)$, respectively.
\end{assumption}

{The class of phase-type distributions include various distributions of theoretical and practical interests. For example, if we choose the parameters~$p=1$, $T = -\beta$, and~$\phi = 1$, then the phase-type distribution $(\phi, T)$ equals an exponential distribution with mean $1/\beta$. A phase-type distribution can also represent various classes of distributions including the Erlang, Coxian, and hyper-exponential distributions~\cite{Commault2003}. Furthermore, it is known that the set of phase-type distributions is dense in the set of positive-valued distributions~\cite{Cox1955}. Therefore, it is possible to approximate an arbitrary distribution by a phase-type distribution within any given accuracy. Moreover, there are efficient numerical algorithms for finding the parameters of an approximating phase-type distribution~\cite{Asmussen1996}. Hence, under  Assumption~\ref{assm:phaseTYpe}, the GeNeSIS model allows us to efficiently approximate realistic spreading processes having non-Markovian transmission and recovery distributions.}

\section{Vectorial Representations} \label{sec:vectorRep}

The aim of this section is to introduce a vectorial representation of the GeNeSIS model under Assumption~\ref{assm:phaseTYpe}. We start our exposition by providing a vectorial representation of an arbitrary phase-type distribution (Subsection~\ref{subsec:vectorRepforPH}), and then present a vectorial representation of the GeNeSIS model~(Subsection~\ref{subsec:vectorRepforGSIS}) that shall be used in Section~\ref{sec:analysis} for analyzing the decay rate of the spreading model.

\subsection{Vectorial Representation of Phase-Type Distributions}\label{subsec:vectorRepforPH}

In what follows, we use the following notation: For $m,m'\in [p]$, let $E_{mm'}$ denote the $p\times p$ matrix whose entries are all zeros, except for its $(m,m')$-th entry being one. Also, let $e_m$ denote the $m$-th vector of the canonical basis in $\mathbb{R}^p$ (i.e., all the entries of~$e_m$ are zeros, except for the $m$-th entry being one). Finally, given a probability distribution $\phi$ on $[p]$, we say that an $\mathbb{R}^p$-valued random variable~$x$ follows the distribution~$\phi$, denoted by $x\sim \phi$, if
\begin{equation*}
P(x =
e_m) = \phi_m
\end{equation*}
for every $m\in [p]$.

{By identifying the state space of the underlying Markov process by the set of vectors~$\{e_1, \dotsc, e_p, 0\}\subset \mathbb{R}^p$, the following proposition allows us to represent the phase-type distribution as the  \emph{exit time} (see \cite[p.~117]{Oksendal2003}) of a vectorial stochastic differential equation:} 

\begin{proposition}\label{prop:PHdist}
Let $(\phi, T)$ be a phase-type distribution having $p+1$ phases (i.e., $T \in
\mathbb{R}^{p\times p}$) and define the vector~$b \in \mathbb{R}^p$ as in
\eqref{eq:inf-gen}. Consider the $\mathbb{R}^p$-valued stochastic differential
equation
\begin{equation}\label{eq:PHnonrenewal}
dx 
= 
\sum_{m=1}^p
\sum_{m'=1}^p(E_{m'm} - E_{mm}) x\,dN_{T_{mm'}} 
- 
\sum_{m=1}^p E_{mm} x \,dN_{b_m},
\end{equation}
with random initial condition~$x(0) \sim \phi$. Then, the random variable
\begin{equation}\label{eq:def_rho}
\begin{aligned}
\rho
&= 
\min\{t > 0\colon x(t) = 0\}
\\
&=
\min\{t > 0\colon \exists\,m\in [p] \text{ such that }x(t^-) = e_m,
\\&\hspace{3cm}
\text{ and }N_{b_m} \text{ jumps at time~$t$}\}
\end{aligned}
\end{equation}
follows $(\phi, T)$.
\end{proposition}

\begin{proof}
{A detailed investigation of the differential equation~\eqref{eq:PHnonrenewal} shows that} the solution~$x$ of \eqref{eq:PHnonrenewal} is a Markov process with state space~$\{e_1, \dotsc, e_p, 0\}\subset \mathbb{R}^p$ and infinitesimal generator given by~\eqref{eq:inf-gen}. {This fact specifically shows that} the second equality in~\eqref{eq:def_rho} is true. Furthermore, since $x(0)$ follows the probability distribution~$\phi$, the time to absorption of the stochastic process~$x$ into the absorbing state~$0$ follows the phase-type distribution~$(\phi, T)$. This completes the proof of the proposition.
\end{proof}

{The stochastic differential equation~\eqref{eq:PHnonrenewal} shall be used for describing phase-type recovery events in the proof of our first main result (Theorem~\ref{thm:main:stability}). On the right-hand side of the stochastic differential equation~\eqref{eq:PHnonrenewal}, the first term represents the transitions between non-absorbing states, while the second term represents the transitions into the absorbing state.}

{In order to derive stochastic differential equations for transmission events, we appropriately modify the second term in the stochastic differential equation~\eqref{eq:PHnonrenewal} and derive a vectorial representation for \emph{renewal sequences} (see, e.g., \cite[Chapter~9]{Cinlar1975a}) whose inter-renewal times follow a phase-type distribution:}

\begin{proposition}\label{prop:PHrenewal}
Let $(\phi, T)$ be a phase-type distribution. Let $e_\phi = \{e_\phi(t)\}_{t\geq
0}$ be independent and identically distributed random variables such that $e_\phi(t)$ follows the
distribution~$\phi$ for all $t\geq 0$. Consider the $\mathbb{R}^p$-valued
stochastic differential equation
\begin{equation}\label{eq:PHrenewal}
\begin{multlined}[.8\linewidth]
dx 
= 
\sum_{m=1}^p \sum_{m'=1}^p
(E_{m'm} - E_{mm}) x\,dN_{T_{mm'}} 
\\
+ 
\sum_{m=1}^p 
(e_\phi(t) e_m^\top - E_{mm}) x \,dN_{b_m}
\end{multlined}
\end{equation}
with a random initial state $x(0)$ following the distribution $\phi$. Define
$\tau_0 = 0$ and let $0 < \tau_1 < \tau_2 < \cdots$ be the (random) times at
which $x(t^-) = e_m$ and the counter $N_{b_m}$ jumps for some $m\in [p]$. Then,
the increments~$\{\tau_k-\tau_{k-1}\}_{k\in \mathbb{N}}$ are independent and
identically distributed random variables following $(\phi, T)$.
\end{proposition}

\begin{proof}
Since the stochastic differential equation given in the proposition is equivalent to the one in Proposition~\ref{prop:PHdist} on the time interval $[0, \tau_1)$, the random variable~$\tau_1$ has the same probability distribution as the random variable~$\rho$ defined in~\eqref{eq:def_rho} and, therefore, follows the phase-type distribution~$(\phi, T)$ by Proposition~\ref{prop:PHdist}. Furthermore, by the definition of $\tau_1$, there exists an $m\in[p]$ such that $x(\tau_1) = (e_\phi(\tau_1) e_m^\top - E_{mm})x(\tau_1^-) + x(\tau_1^-) = e_\phi(\tau_1)$, which follows $\phi$. Therefore, by the same argument as above, we see that the random increment $\tau_2-\tau_1$ also follows $(\phi, T)$. An induction completes the proof.
\end{proof}

{Using Propositions~\ref{prop:PHdist} and~\ref{prop:PHrenewal} as the machinery for describing phase-type recovery and transmission events, in the next subsection we present a set of vectorial stochastic differential equations for describing the whole GeNeSIS spreading model.}

\subsection{Vectorial Representation of the Generalized SIS Model
}\label{subsec:vectorRepforGSIS}

In this subsection, we use Propositions~\ref{prop:PHdist} and 
\ref{prop:PHrenewal} to provide a vectorial representation of the GeNeSIS model
under Assumption~\ref{assm:phaseTYpe}. Let $A = [a_{ij}]_{i,j}$ be the adjacency
matrix of the graph~$\mathcal G$. Define the vectors $b$ and $d$ by \eqref{eq:inf-gen} and 
\begin{equation*}
d = -R \onev_q \in \mathbb{R}^q, 
\end{equation*}
respectively. For $\ell, \ell'\in [q]$, let $F_{\ell\ell'}$ denote the $q\times q$ matrix whose entries are all zeros, except for its $(\ell,\ell')$-th entry being one. Also, let $f_\ell$ denote the $\ell$-th vector in the canonical basis of $\mathbb{R}^q$. Finally, for $i, j\in[n]$ and~$\gamma>0$, we let $N^{ij}_\gamma$ and~$N^i_\gamma$ denote independent Poisson counters with rate $\gamma$. The next theorem is the first main result of this paper:

\begin{theorem}\label{thm:vectorRepresentation}
For each $i\in [n]$ and~$j\in\mathcal N_i$, let $e_\phi^{ji} = \{e_\phi^{ji}(t)
\}_{t\geq 0}$ and~$f_\psi^i = \{f_\psi^i(t)\}_{t\geq 0}$ be independent and identically distributed random variables such that 
\begin{equation*}
e_\phi^{ji}(t)\sim \phi
, 
\quad 
f^i_\psi(t)\sim \psi
\end{equation*}
for all $t\geq 0$. Let $x^{ji}$ and~$y^i$ be, respectively, the $\mathbb{R}^p$-
and~$\mathbb{R}^q$-valued stochastic processes satisfying the following
stochastic differential equations:
\begin{align}
dx^{ji} 
&=
\sum_{m=1}^p \sum_{m'=1}^p
(E_{m'm} - E_{mm}) x^{ji}\,dN_{T_{mm'}}^{ji} 
\notag
\\
&\myindent+
\sum_{m=1}^p (e_\phi^{ji} e_m^\top - E_{mm}) x^{ji}\,dN_{b_m}^{ji} 
\notag
\\
&\myindent-
x^{ji}\sum_{\ell=1}^q y_\ell^i \,dN_{d_\ell}^i 
+ 
e_\phi^{ji}(1-\onev^\top_q y^i)
\sum_{k=1}^n a_{ik} \sum_{m=1}^p x_m^{ik} \,dN_{b_m}^{ik}, 
\label{eq:SDE:transmission}
\\
dy^i 
&= 
\sum_{\ell=1}^q\sum_{\ell'=1}^q
(F_{\ell'\ell} - F_{\ell\ell})y^i\,dN_{R_{\ell \ell'}}^i 
\notag
\\
&\myindent
- 
y^i \sum_{\ell=1}^q y_\ell^i \,dN_{d_\ell}^i  
+
f_\psi^i (1-\onev^\top_q y^i)
\sum_{k=1}^n a_{ik} \sum_{m=1}^p x_m^{ik} \,dN_{b_m}^{ik}, 
\label{eq:SDE:recovery}
\end{align}
where for an initially infected subset $\mathcal V_0 \subset[n]$, the initial conditions satisfy
\begin{equation}\label{eq:initCond}
\begin{cases}
x^{ji}(0)\sim \phi, \ y^{i}(0)\sim \psi, &\text{if $i\in \mathcal V_0$}, 
\\
x^{ji}(0)
=0_p,\ y^i(0) = 0_q, &\text{otherwise.}
\end{cases}
\end{equation} 
Then, the generalized networked SIS model in Definition \ref{defn:GSIS} with
transmission and recovery times following, respectively, phase-type
distributions~$(\phi, T)$ and~$(\psi, R)$ can be equivalently described as the
family of stochastic processes $z = \{z_i\}_{i=1}^n$, where
\begin{equation}\label{eq:def:z}
z_i(t) = \onev^\top_q y^i(t)
\end{equation}
for all $t\geq 0$ and~$i\in [n]$.
\end{theorem}

{The representations of the GeNeSIS model as the set of stochastic differential equations~\eqref{eq:SDE:transmission} and~\eqref{eq:SDE:recovery} allows us to analyze the model via symbolic computations, as will be illustrated in Section~\ref{sec:analysis}. Before proceeding to the proof of Theorem~\ref{thm:vectorRepresentation}, we provide an intuitive explanation of the theorem. As is shown in Corollary~\ref{cor:} below, the variable $x^{ji}$ is related to spread of the infection from an infected node~$i$ to a susceptible node~$j$, while $y^i$ controls the recovery process of node~$i$. Specifically, on the right-hand side of the  differential equation~\eqref{eq:SDE:transmission}, the first two terms have the same structure as in \eqref{eq:PHrenewal} and correspond to renewal sequences of transmissions. The third and fourth terms represent the recovery and infection of node~$i$, respectively. Similarly, on the right-hand side of the  differential equation~\eqref{eq:SDE:recovery}, the first two terms correspond to the phase-type recovery and have almost the same structure as in~\eqref{eq:PHnonrenewal}, while the remaining last term is related to the infection of the node~$i$.}

{In order to prove Theorem~\ref{thm:vectorRepresentation}, we first state the following lemma.} 

\begin{lemma}\label{lem:basicProperties}
The following statements are true for all $i\in [n]$, $j\in \mathcal N_i$, and
$t\geq 0$: 
\begin{enumerate}
\item $\onev^\top_p x^{ji}(t) = \onev^\top_q y^i(t)$. 
\item $x^{ji}(t) \in 
\{0_p, e_1, \dotsc, e_p \}$. 
\item $y^i(t) \in \{0_q, f_1, \dotsc, f_q\}$.
\end{enumerate}
\end{lemma}

\begin{proof}
To prove the first statement, fix $i\in [n]$ and~$j\in\mathcal N_i$, and let
\begin{equation*}
\epsilon = \onev^\top_p x^{ji} - \onev^\top_q y^i.
\end{equation*}
From \eqref{eq:SDE:transmission} and~\eqref{eq:SDE:recovery}, we obtain that
\begin{equation*}
d\epsilon = \epsilon \sum_{\ell=1}^qy^i_\ell \,dN^i_{d_\ell}.
\end{equation*}
This equation implies that $\epsilon$ is constant over~$[0, \infty)$ because
$\epsilon(0) = 0$. Therefore, we have $\epsilon(t) = \epsilon(0) =0$ for every
$t\geq 0$, completing the proof of the first statement.

Let us prove the second and third statements. Notice that $x^{ji}$ and~$y^i$ are
piecewise constant since they are the solutions of the stochastic differential
equations~\eqref{eq:SDE:transmission} and~\eqref{eq:SDE:recovery}. Moreover,
their values can change only when one of the Poisson counters in the stochastic
differential equations jumps. Finally, if we let 
\begin{equation*}
\begin{aligned}
U &= \{0_p, e_1, \dotsc, e_p \},\\
V &= \{0_q, f_1, \dotsc,
f_q\}, 
\end{aligned}
\end{equation*}
then, \eqref{eq:initCond} shows $x^{ji}(0)
\in U$ and~$y^i(0)\in V$. Therefore, it is sufficient to show that the jump of any Poisson counter
leaves the sets $U$ and~$V$ invariant.

Let $t > 0$ be arbitrary and assume that $x^{ji}(t^-)\in U$ and~$y^i(t^-) \in V$. The stochastic differential equations~\eqref{eq:SDE:transmission} and~\eqref{eq:SDE:recovery} have the following five different types of Poisson counters: ${N_{b_m}^{ik}}$, ${N_{d_\ell}^i}$, $N_{T_{m\ssp m}}^{ji}$, ${N_{R_{\ell\ssp \ell}}^i}$, and ${N_{b_m}^{ji}}$. {Careful investigations of the stochastic differential equations~\eqref{eq:SDE:transmission} and~\eqref{eq:SDE:recovery} show} that any of these counters leave the sets $U$ and~$V$ invariant {as follows. For example, when} $N_{b_m}^{ik}$ jumps at time~$t$, we have that
\begin{equation}\label{eq:countercase4:1}
\begin{aligned}
x^{ji}(t) 
=&\,
x^{ji}(t^-)  + e_\phi^{ji}(t)(1-\onev^\top_q y^i(t^-)) a_{ik} x_m^{ik} (t^-)
\\
=&\,
\begin{cases}
x^{ji}(t^-), & \text{if $\onev^\top_q y^i(t^-)=1$, $a_{ik}=0$, or $x_m^{ik}(t^-) =0$}, 
\\
e_\phi^{ji}(t),&\text{otherwise}.
\end{cases}
\end{aligned}
\end{equation}
Also, 
\begin{equation}\label{eq:countercase4:2}
\begin{aligned}
y^i(t) 
&=
y^i(t^-)  + f_\psi^i(t)(1-\onev^\top_q y^i(t^-)) a_{ik} x_m^{ik}(t^-)
\\
&=
\begin{cases}
y^i(t^-), & \text{if $\onev^\top_q y^i(t^-)=1$, $a_{ik}=0$, or $x_m^{ik}(t^-) =0$}, 
\\
f_\psi^i(t),&\text{otherwise}.
\end{cases}
\end{aligned}
\end{equation}
Therefore, this jump leaves $U$ and~$V$ invariant.
Similarly, when $N_{d_\ell}^i$ jumps at time~$t$, we obtain 
\begin{equation}\label{eq:countercase5:1}
\begin{aligned}
x^{ji}(t)
&=
x^{ji}(t^-) - x^{ji}(t^-) y_\ell^i(t^-)
\\
&=
\begin{cases}
0, & \text{if $y_\ell^i(t^-) = 1$}, 
\\
x^{ji}(t^-),&\text{otherwise}, 
\end{cases}
\end{aligned}
\end{equation}
and hence $x^{ji}(t)\in U$. We also have $y^i(t)\in V$ because
\begin{equation}\label{eq:countercase5:2}
\begin{aligned}
y^i(t)
&=
y^i(t^-) - y^i(t^-) y_\ell^i(t^-)
\\
&=
\begin{cases}
0, & \text{if $y_\ell^i(t^-) = 1$}, 
\\
y^i(t^-), &\text{otherwise}.
\end{cases}
\end{aligned}
\end{equation}
{The cases of other counters can be analyzed in a similar manner and, therefore, the proofs are omitted.} %
\end{proof}

{The next corollary of Lemma~\ref{lem:basicProperties} clarifies the roles of the variables $x^{ji}$, $y^i$ and the various Poisson counters in the stochastic differential equations~\eqref{eq:SDE:transmission} and \eqref{eq:SDE:recovery}.}

\begin{corollary}\label{cor:}
The following statements are true for every $i\in [n]$, $j\in \mathcal N_i$, and~$t\geq 0$:
\begin{enumerate}
\item Node~$i$ attempts to infect node~$j$ at time~$t$, if and only if,
$x_m^{ji}(t^-) = 1$ and~$N_{b_m}^{ji}$ jumps at time~$t$ for some
$m\in [p]$.

\item Node~$i$ becomes susceptible at time~$t$, if and only if, $y_\ell^i(t^-) =
1$ and~$N_{d_\ell}^i$ jumps at time~$t$ for some $\ell \in [q]$.
\end{enumerate}
\end{corollary}

\begin{proof}
From the proof of Lemma~\ref{lem:basicProperties}, the value of $z_i= \onev^\top_p x^{ji} = \onev^\top_q y^i$ can change from zero to one {when and only when a counter $N_{b_m}^{ik}$ jumps. Therefore, from equations~\eqref{eq:countercase4:1} and~\eqref{eq:countercase4:2}}, we see that node~$i$ becomes infected at time $t\geq 0$, if and only if, $z_i(t^-) = 0$, node~$i$ is adjacent to node~$k$, $x^{ik}_m(t^-)=1$, and~$N_{b_m}^{ik}$ jumps at time~$t$ for some $m \in [p]$. Therefore, node~$k$ attempts to infect node~$i$ at time~$t$ if and only if $x^{ik}_m(t^-)=1$ and~$N_{b_m}^{ik}$ jumps at time~$t$ for some $m \in [p]$, as desired. 

Similarly, we see that the value of $z^i$ can change from one to zero {when and only when a counter $N_{d_\ell}^i$ jumps. Therefore,  equations~\eqref{eq:countercase5:1} and~\eqref{eq:countercase5:2}} imply that node~$i$ becomes susceptible at time~$t$ if and only if $y^i_\ell(t^-)=1$ and the counter~$N_{d_\ell}^i$ jumps at time~$t$ for some $\ell \in [q]$, showing the second assertion of the corollary.
\end{proof}

Now we are ready to prove 
Theorem~\ref{thm:vectorRepresentation}.

\begin{proof}[Proof of Theorem~\ref{thm:vectorRepresentation}]
We prove that the family of stochastic processes $z = \{z_i\}_{i=1}^n$, defined by \eqref{eq:SDE:transmission}--\eqref{eq:def:z}, satisfies items a)-c) in Definition~\ref{defn:GSIS}, as well as Assumption~\ref{assm:phaseTYpe}. Item~a) is true by the given initial conditions. Let us prove item~b). Assume that node~$i$ becomes infected at time~$t\geq0$; hence, we have that either $t=0$ or $t>0$. In this proof, we only consider the case $t=0$, as the other case can be proved in a similar way. Let $\rho$ be the earliest time at which $y^i(\rho)=0_q$. Then, on the time interval~$[0, \rho)$, the stochastic differential equation \eqref{eq:SDE:recovery} is equivalent to
\begin{equation*}
dy^i 
= 
\sum_{\ell = 1}^q  \sum_{\ell'=1}^q
(F_{\ell'\ell} -
F_{\ell\ell})y^i\,dN_{R_{\ell \ssp \ell}}^i - \sum_{\ell=1}^q F_{\ell\ell}y^i
\,dN_{d_\ell}^i
\end{equation*}
since $y^i y_\ell^i = F_{\ell\ell} y^i$. Moreover, the vector $y^i(0)$ follows the distribution~$\psi$. Therefore, by Proposition~\ref{prop:PHdist}, we see that $\rho$ follows the phase-type distribution~$(\psi, R)$, proving item~b) in Definition~\ref{defn:GSIS} under Assumption \ref{assm:phaseTYpe} for $t=0$. {A similar discussion using Proposition~\ref{prop:PHrenewal} shows that} the family of stochastic processes~$z$ satisfies item~c) in Definition~\ref{defn:GSIS} under Assumption~\ref{assm:phaseTYpe}. This completes the proof of the theorem.
\end{proof}

Before presenting our analysis of the decay rate in the next section, we state another corollary of Lemma~\ref{lem:basicProperties}:

\begin{corollary}\label{cor:xi}
The family of stochastic processes 
\begin{equation}\label{eq:def:xi}
\xi = \{x^{ji}, y^{i}\}_{i\in[n],
j\in\mathcal N_i}, 
\end{equation}
where $x^{ji}$ and~$y^i$ are the solutions of the stochastic differential equations~\eqref{eq:SDE:transmission} and~\eqref{eq:SDE:recovery}, is a Markov process. Moreover, the states of $\xi$ are all transient except for the absorbing state~$\alpha$ at which $x^{ji}=0$ and~$y^{i}=0$ for all $i\in[n]$ and~$j\in\mathcal N_i$. Furthermore, the size of the state space of~$\xi$ equals
\begin{equation*}
N_\xi = \prod_{i=1}^n (1+p^{\abs{\mathcal N_i}} q),
\end{equation*}
where $\abs{\mathcal N_i}$ denotes the size of the neighbor set~$\mathcal N_i$.
\end{corollary}

\begin{proof}
A careful investigation of the proof of Lemma~\ref{lem:basicProperties} shows the first and the second claims. The third claim is an immediate consequence from the constraint $\onev^\top_p x^{ji}(t) = \onev^\top_q y^i(t)$ that was proved in Lemma~\ref{lem:basicProperties}.
\end{proof}

\section{{Decay Rate Analysis}}\label{sec:analysis}

{In this section, we use our previous results to bound the exponential decay rate of the GeNeSIS model under Assumption~\ref{assm:phaseTYpe}. We begin by presenting a characterization of the decay rate in terms of the eigenvalues of a matrix whose size grows exponentially with respect to the model parameters. To overcome the	 difficulty of computing the eigenvalues of a very large matrix, we then present an alternative bound on the decay rate based on the representation of the GeNeSIS model as the stochastic differential equations~\eqref{eq:SDE:transmission} and~\eqref{eq:SDE:recovery}.}

{Throughout this section, we consider the GeNeSIS model with transmission and recovery times following, respectively, phase-type distributions $(\phi, T)$ and~$(\psi, R)$ (i.e., satisfying~Assumption~\ref{assm:phaseTYpe}). The following proposition illustrates the computational difficulty in computing the exponential decay rate of the GeNeSIS model:}

\begin{proposition}\label{prop:exact}
Let $Q \in \mathbb{R}^{N_\xi \times N_\xi}$ be the transition rate matrix of the Markov process~$\xi$ (see \eqref{eq:def:xi} for the definition of $\xi$), and let $r < 0$ be the maximum real part of the non-zero eigenvalues of~$Q$. Then, the exponential decay rate of the GeNeSIS model is given by
\begin{equation*}
\lambda = -r. 
\end{equation*}
\end{proposition}

\begin{proof}
Since $\xi(t) \neq \alpha$ if and only if at least one node is infected at time
$t$, we have the inequality
\begin{equation*}
E[z_i(t)] \leq  P(\xi(t) \neq \alpha)\leq
\sum_{i=1}^n E[z_i(t)]
\end{equation*}
for all $i\in [n]$ and~$t\geq 0$. Therefore, the exponential decay rate of the GeNeSIS model is determined by that of the function~$P(\xi(t) \neq \alpha)$ of $t$. By Corollary~\ref{cor:xi}, a basic argument on Markov processes (see \cite{VanMieghem2009a} for the case of exponential transmission and recover times) shows that \emph{i}) if $\lambda <-r$, there exists a constant $C>0$ such that $P(\xi(t) \neq \alpha) \leq Ce^{-\lambda t}$ for all $t$ and any initial state $\xi(0)$, and \emph{ii}) if $\lambda>-r$, there exists an initial state of $\xi$ such that the function~$P(\xi(t) \neq \alpha)$ cannot be bounded from above by the exponential function~$C e^{-\lambda t}$ for any value of~$C$. These observations immediately prove the claim of the proposition.
\end{proof}

{Proposition~\ref{prop:exact} yields the exact value of the decay rate, since the proposition uses the transition rate matrix of the whole Markov process~$\xi$ that exactly describes the GeNeSIS model. However, Proposition~\ref{prop:exact} is not easily applicable in practice because the dimension~$N_\xi$ of the matrix~$Q$ grows exponentially as   
\begin{equation}\label{eq:Nx->}
N_\xi \geq \prod_{i=1}^n ( p^{\abs{\mathcal N_i}} q) = p^{2m} q^n, 
\end{equation}
where $m$ denotes the number of the edges in the network. The following theorem overcomes this computational difficulty by providing a lower bound on the growth rate in terms of the eigenvalues of a matrix whose size grows \emph{linearly} with respect to the parameters in the GeNeSIS model.}

\begin{theorem}\label{thm:main:stability}
Define the $(npq)\times (npq)$ matrix
\begin{equation}\label{eq:matcalA}
\mathcal A = (\phi b^\top) \otimes  A\otimes  (\psi \onev_q^\top) + I_{np}\otimes R^\top + (T^\top + \phi b^\top) \otimes I_{nq} ,
\end{equation}
where $A$ is the adjacency matrix of the graph~$\mathcal G$ and the vector~$b$
is defined in \eqref{eq:inf-gen}. Then, 
\begin{equation}\label{eq:mainIneq}
\lambda \geq {-\eta}(\mathcal A),
\end{equation}
where $\eta(\mathcal A)$ is the spectral abscissa of $\mathcal A$.
\end{theorem}

Before providing a proof of Theorem~\ref{thm:main:stability}, we below present a series of corollaries of the theorem. The proofs of the corollaries are straightforward and, therefore, omitted. The first corollary gives a bound on the decay rate of the GeNeSIS model with exponential transmission times and phase-type recovery times.

\begin{corollary}\label{cor:expTrans}
Assume that the transmission times follow an exponential distribution with mean $1/\beta$. Define the $(nq)\times (nq)$ matrix~
\begin{equation*}
\mathcal A_\beta = \beta A\otimes (\psi \onev^\top_q) + I_n \otimes R^\top
\end{equation*}
Then, the decay rate satisfies $\lambda \geq {-\eta}(\mathcal A_\beta)$. 
\end{corollary}

The next corollary deals with the dual case with phase-type transmission times and exponential recovery times.

\begin{corollary}\label{cor:expRec}
Assume that the recovery times follow an exponential distribution with mean $1/\delta$. Define the $(np)\times (np)$ matrix
\begin{equation*}
\mathcal A_\delta = (\phi b^\top) \otimes A + (T^\top+\phi b^\top)\otimes I_n -\delta I_{np}.
\end{equation*}
Then, the decay rate satisfies $\lambda \geq {-\eta}(\mathcal A_\delta)$. 
\end{corollary}

As the special case of Theorem~\ref{thm:main:stability}, as well as  Corollaries~\ref{cor:expTrans} and~\ref{cor:expRec}, we can prove the following bound on the decay rate of the standard SIS model:

\begin{corollary}[{\cite{Ganesh2005,Preciado2014}}] \label{cor:expTransRec}
Assume that the transmission and recovery times follow exponential distributions with means $1/\beta$ and $1/\delta$, respectively. Define the $n\times n$ matrix
\begin{equation*}
\mathcal A_{\beta, \delta} = \beta A - \delta I_n.
\end{equation*}
Then, the decay rate satisfies $\lambda \geq {-\eta}(\mathcal A_{\beta, \delta})$. 
\end{corollary}

The above corollaries suggest an intuitive understanding of the terms in the matrix~$\mathcal A$ defined in \eqref{eq:matcalA}. Comparing the expressions of the matrices~$\mathcal A_\delta$ and~$\mathcal A_{\beta, \delta}$, we see that the role of the exponential transmission rate~$\beta$ is played by the $p\times p$ matrix~$\phi b^\top$ in the case of phase-type transmission times. On the other hand, the second term of the matrix~$\mathcal A_\delta$, namely, $(T^\top+\phi b^\top)\otimes I_n$, can be understood as a correction term that arises independently of the topology of the network. Similarly, comparing the matrices~$\mathcal A_\beta$ and~$\mathcal A_{\beta, \delta}$, we see that the second term~$I_n\otimes R^\top$ of the matrix $\mathcal A_\beta$ represents the effect of phase-type recoveries. On the other hand, we can understand the matrix~$\psi\onev^\top_q$ as a correction term resulting from using phase-type recovery times.

We now give the proof of Theorem~\ref{thm:main:stability}.    

\begin{proof}[Proof of Theorem~\ref{thm:main:stability}]
Combining equations~\eqref{eq:SDE:transmission} and~\eqref{eq:SDE:recovery}, we
obtain the following $\mathbb{R}^{p+q}$\nobreakdash-valued stochastic
differential equation
\begin{equation*}
\renewcommand*{\arraystretch}{.75}
\begin{aligned}
d\begin{bmatrix}
x^{ji}\\y^i
\end{bmatrix} 
&= 
\sum_{m=1}^p \sum_{m'=1}^p
\begin{bmatrix}
(E_{m'm} - E_{mm}) x^{ji}\\0_{nq}
\end{bmatrix}dN_{T_{mm'}}^{ji}
\notag\\&\myindent
+ \sum_{m=1}^p
\begin{bmatrix}
(e_\phi^{ji} e_m^\top - E_{mm}) x^{ji}\\0_{nq}
\end{bmatrix}dN_{b_m}^{ji} 
\\
&\myindent+ 
\sum_{\ell=1}^q
\sum_{\ell'=1}^q
\begin{bmatrix}
0_{np} \\ (F_{\ell'\ell} - F_{\ell\ell})y^i
\end{bmatrix}dN_{R_{\ell \ell'}}^i
\\
&\myindent+ 
\sum_{\ell=1}^q 
\begin{bmatrix}
-x^{ji} y_\ell^i
\\ 
-y^i y_\ell^i
\end{bmatrix}dN_{d_\ell}^i
\\&\myindent+ 
\sum_{k=1}^n\sum_{m=1}^p
a_{ik}\begin{bmatrix}
e_\phi^{ji}(1-\onev^\top_qy^i)x_m^{ik}
\\
f_\psi^i (1-\onev^\top_qy^i) x_m^{ik}
\end{bmatrix}
dN_{b_m}^{ik}.
\end{aligned}
\end{equation*}
Then, we apply It\^o's formula in Lemma~\ref{lem:Ito} using the function 
\begin{equation*}
g\left(
\begin{bmatrix}
x^{ji}\\y^i
\end{bmatrix}
\right) = w^{ji} = x^{ji} \otimes y^i
\end{equation*}
to obtain the following $\mathbb{R}^{p q}$\nobreakdash-valued stochastic differential equation (after tedious, but simple, algebraic manipulations)
\begin{align}
dw^{ji} 
&=
\sum_{m=1}^p \sum_{m'=1}^p
\bigl((E_{m'm} - E_{mm}) \otimes I_q\bigr) w^{ji}\,dN_{T_{mm'}}^{ji}
\notag
\\
&\myindent
+ 
\sum_{m=1}^p
\bigl((e_\phi^{ji} e_m^\top - E_{mm}) \otimes I_q\bigr) w^{ji}
\,dN_{b_m}^{ji}  
\notag
\\
&\myindent+ 
\sum_{\ell=1}^q \sum_{\ell'=1}^q
\bigl(I_p \otimes (F_{\ell'\ell} - F_{\ell\ell})\bigr) w^{ji}
\, dN_{R_{\ell \ell'}}^i \notag 
\\
&\myindent+ 
\sum_{\ell=1}^q 
(-y_\ell^i w^{ji})\,dN_{d_\ell}^i 
\notag\\&\myindent
+\sum_{k=1}^n\sum_{m=1}^p
a_{ik}(e_\phi^{ji} \otimes f_\psi^i) (1-\onev^\top_{pq} w^{ji}) (e_m^\top \otimes \onev_q^\top) w^{ik}\,dN_{b_m}^{ik}.
\label{eq:dwji}
\end{align}
For brevity, we omit the details of this derivation. Define 
\begin{equation*}
\omega_{ji}(t) = E[w^{ji}(t)]
\end{equation*}
for $t\geq 0$. Then, using Lemma~\ref{lem:expectation}, from \eqref{eq:dwji} we can derive the $\mathbb{R}^{p q}$-valued differential equation
\begin{align}
\frac{d\omega_{ji}}{dt}
&=
\Bigl[ (T^\top + B) \otimes I_q
+ 
(\phi b^\top - B)\otimes I_q  
\notag
\\ 
&\myindent
+\bigl(I_p\otimes (R^\top + D)\bigr)\Bigr]\omega_{ji} -(I_p\otimes D)\omega_{ji} 
\notag 
\\
&\myindent 
+(\phi \otimes \psi) (e_m^\top \otimes \onev_q^\top) \sum_{m=1}^p b_m\sum_{k=1}^n
(a_{ik}\omega_{ik})  - \epsilon_{ji},
\label{eq:domegaji/dt}
\end{align}
where $B$ and~$D$ are the $n\times n$ diagonal matrices having the diagonals $b_1,\dotsc,b_n$ and~$d_1,\dotsc,d_n$, respectively, and the $\mathbb{R}^{pq}$\nobreakdash-valued function $\epsilon_{ji}$ is defined by
\begin{equation}\label{eq:epsilon}
\epsilon_{ji}(t) = \sum_{k=1}^n\sum_{m=1}^p b_m a_{ik}(\phi \otimes \psi)
E[\onev^\top_{pq} w^{ji}(t)(e_m^\top \otimes \onev_q^\top) w^{ik}(t)]
\end{equation}
for every $t\geq 0$.

\begin{figure*}[tb]
\centering 
\includegraphics[width=.975\linewidth]{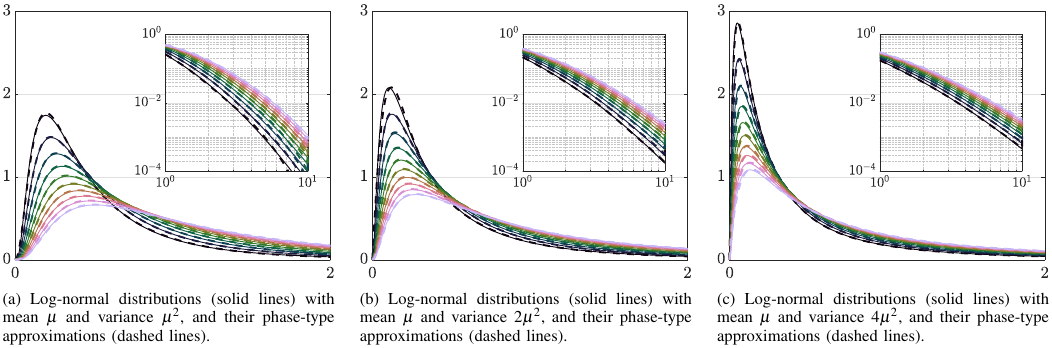}
\caption{Approximations by phase-type distributions. Solid: Probability density functions of the log-normal
distributions for $\mu$ in the range $[0.5:0.1:1.5]$. Dashed:
Probability density functions of phase-type distributions. The
darker the colors of the plots, the smaller the values of~$\mu$.}
\label{fig:fitting}
\end{figure*}

Now, for every $i\in [n]$, we take an arbitrary $j_i \in \mathcal N_i$. We then
define the function $\omega_i$ by $\omega_i= \omega_{j_ii}$, as well as the 
$\mathbb{R}^{n p q}$-valued function
\begin{equation*}
\omega = \col(\omega_1, \dotsc, \omega_n).
\end{equation*}
Notice that, from \eqref{eq:dwji}, {for  each  $i  \in [n]$,} all the stochastic processes in the set~$\{w^{ji}\}_{j\in\mathcal N_i}$ follow the same stochastic differential equation and, therefore, present the same probability distribution. This implies that $\omega_i = \omega_{ji}$ for every $j\in \mathcal N_i$. Therefore, we can rewrite the last summation appearing in \eqref{eq:domegaji/dt} as $\sum_{k=1}^n a_{ik}\omega_{ik} = \sum_{k=1}^n a_{ik}\omega_{k} = (A_i \otimes I_{pq})\omega$, where $A_i$ denotes the $i$-th row of the adjacency matrix~$A$. Then, from (11), it follows that
\begin{equation*}
\begin{multlined}[.8\linewidth]
\frac{d\omega_i}{dt} = \left( T^\top  \otimes I_q + I_p\otimes R^\top + (\phi
b^\top)\otimes I_q\right)\omega_i 
\\
+\left(A_i \otimes (\phi b^\top) \otimes
(\psi\onev^\top_q)\right)\omega-\epsilon_{j_ii}.
\end{multlined}
\end{equation*}
Defining 
\begin{equation*}
\epsilon =
\col(\epsilon_{j_11}, \dotsc, \epsilon_{j_nn}),
\end{equation*}
we obtain the differential
equation 
\begin{equation*}
\frac{d\omega}{dt}=\mathcal A' \omega -\epsilon,
\end{equation*}
for the matrix $\mathcal A' = I_n \otimes ( T^\top\otimes I_q + I_p \otimes R^\top + (\phi b^\top)\otimes I_q ) + A\otimes (\phi b^\top)\otimes (\psi \onev^\top_q)$.

Since $\mathcal A'$ is Metzler, we have that $e^{\mathcal A't} \geq 0$ for every $t\geq 0$. Also, since both $x^{ji}(t)$ and~$y^i(t)$ are nonnegative for all $i\in[n]$, $j\in\mathcal N_i$, and~$t\geq 0$, we have that $\epsilon(t) \geq 0$ for every $t\geq 0$. Therefore, it follows that
\begin{equation*}
\omega(t) = e^{\mathcal A' t}\omega(0) - \int_0^t e^{\mathcal A'(t-\tau)}\epsilon(\tau)\,d\tau \leq e^{\mathcal A' t}\omega(0).
\end{equation*}
This inequality implies that $\omega(t)$ converges to zero as $t\to\infty$ with a decay rate at least $-\eta(\mathcal A')$ since $\omega(t)\geq 0$ for all $t\geq 0$. On the other hand, for each $i\in [n]$ and $j\in \mathcal N_i$, we have
\begin{equation*}
\onev^\top_{pq} w^{ji}(t) 
= 
(\onev^\top_p x^{ji}(t)) (\onev^\top_q y^i(t)) 
= 
\onev^\top_q y^i(t) = z_i(t)
\end{equation*}
from Lemma~\ref{lem:basicProperties}. Therefore, we have $E[z_i(t)] = \onev^\top_{pq}\omega_{ji}(t)$, which shows the exponentially fast convergence of $E[z_i(t)]$ towards zero with a decay rate at least $-\eta(\mathcal A')$. This completes the proof of the inequality~\eqref{eq:mainIneq} since $\mathcal A'$ and $\mathcal A$ are similar.
\end{proof}

\begin{figure*}[tb]
\centering
\includegraphics[width=.975\linewidth]{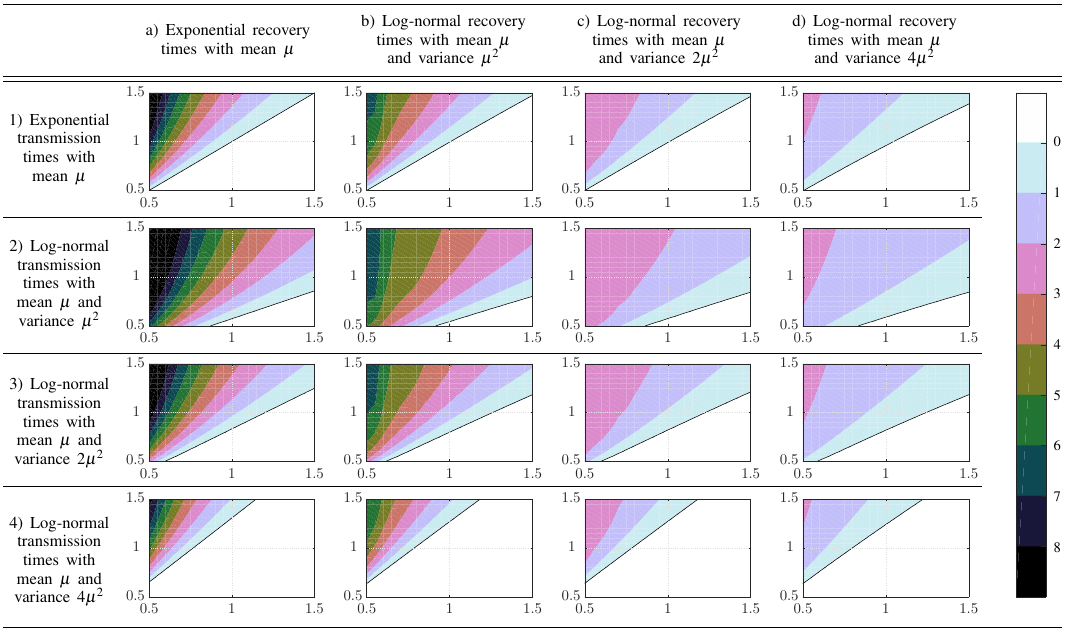}
\caption{Lower bound $-\eta(\mathcal A)$ on the exponential decay rate of
generalized SIS model. Horizontal axes: $\mu$. Vertical axes: $\mu$ divided by the spectral radius of the graph. {The white region corresponds to unstable epidemics; the darker the region, the faster the epidemics converges toward the infection-free equilibrium.}}
\label{fig:table}
\end{figure*}

\begin{remark}
Unlike the necessary and sufficient condition in Proposition~\ref{prop:exact}, the condition in Theorem~\ref{thm:main:stability} is only sufficient. This conservatism arises from ignoring the higher-order term~$\epsilon_{ji}$ in \eqref{eq:epsilon}. The inclusion of these higher-order terms into the analysis (see, e.g., \cite{Ruhi2016a,Ogura2017}) would allow us to reduce the conservatism, at the cost of increasing the dimension of the matrix~$\mathcal A$ .
\end{remark}

\section{Numerical Simulations} \label{sec:sim}

In this section, we illustrate the effectiveness of our results with numerical simulations in a real social network {having $n=247$ nodes and $940$ edges}. {We focus on log-normal transmission and recovery times, which are observed in empirical studies, including information spread on online social networks~\cite{Mieghem2011a,Doerr2013} and human epidemiology~\cite{Limpert2001,NISHIURA2007}. In our simulations, we illustrate the effect of using exponential distributions to model transmission and recovery times that, in reality, follow log-normal distributions. In particular, we analyze how using standard Markovian models with exponential rates induce errors in the computation of the decay rate. Furthermore, we are also interested in how the variances of log-normal distributions, which cannot be incorporated into the standard Markovian model, affect the decay rates. For this purpose,} we use the following four distributions to model transmission and recovery times in the GeNeSIS model: (\emph{i}) the exponential distribution with mean~$\mu$ (and, hence, variance~$\mu^2$); (\emph{ii}) the log-normal distribution with mean~$\mu$ and variance~$\mu^2$; (\emph{iii}) the log-normal distribution with mean~$\mu$ and variance~$2\mu^2$; and (\emph{iv}) the log-normal distribution with mean~$\mu$ and variance~$4\mu^2$.

In order to analyze the {decay rate} of the GeNeSIS model whose transmission and recovery times follow one of these four distributions, we first approximate the three log-normal distributions in (\emph{ii})--(\emph{iv}) using phase-type distributions (as described in Section~\ref{sec:SISmodel}) having $p=10$ phases. Fig.~\ref{fig:fitting} shows the probability density functions of the (exact) log-normal distributions, as well as the fitted phase-type distributions, when the value of the parameter $\mu$ is in the range $[0.5:0.1:1.5]$. {We notice that, since the inequality~\eqref{eq:Nx->} shows that the size of the exact state transition matrix satisfies $N_\xi > 10^{2127}$, it is not practical to use Proposition~\ref{prop:exact} to compute the exact decay rate.}

Using the proposed phase-type distributions, we apply Theorem~\ref{thm:main:stability} to analyze the {decay rate} of the GeNeSIS model when the transmission/recovery times follow one of the four distributions described above. We compute the decay rate~$-\eta(\mathcal A)$ for each one of the possible sixteen combinations of transmission/recovery distributions when $\mu$ is in the range $[0.5:0.05:1.5]$. In the subfigures of Fig.~\ref{fig:table}, we include contour plots of the decay rates for these sixteen cases. {In the figure, the values of the decay rates are indicated by different colors. The darker the color, the faster the epidemics converges to the infection-free equilibrium.} 

We remark that, even though the four distributions used in our simulations have the same mean, the resulting GeNeSIS models exhibit different decay rates. {We can furthermore observe that the variances of the distributions used to model transmission and recovery times can dramatically  affect the decay rate of the GeNeSIS model, as was previously indicated by numerical simulations~\cite{VanMieghem2013} and mean-field approximations~\cite{Cator2013a}.} For a fixed recovery distribution (i.e., for a fixed column in the table in Fig.~\ref{fig:table}), the colored region (indicated by black solid lines) shrinks as we increase the variance of the log-normal distribution modeling the transmission times (see rows~2 to~4 in Fig.~\ref{fig:table}). From this observation, we {see} that the heavier the tail of the transmission distribution, the {slower} the extinction of the spreading process{, as was numerically confirmed in~\cite{VanMieghem2013}.} We also observe that for a fixed transmission distribution (i.e., for a fixed row in the table), the colored region remains almost unaltered as we increase the variance of the distribution modeling the recovery time{, confirming the validity of the mean-field analysis in~\cite{Cator2013a} for the case of exponential transmission times}. We can furthermore observe that the exponential decay rate~$-\eta(\mathcal A)$ increases more abruptly inside the region as we decrease this variance as indicated by steeper gradients. {The above observations illustrate the effectiveness of the proposed framework for analyzing the decay rate of epidemics in networks with non-Poissonian transmission and/or recovery distributions.}

\section{Conclusion} \label{sec:co}

In this paper, we have analyzed the dynamics of an {SIS} model of spreading over arbitrary networks with phase-type transmission and recovery times. Since phase-type distributions form a dense family in the space of positive-valued distributions, our results allow to theoretically analyze arbitrary transmission and recovery times within an arbitrary accuracy. In this context, we have derived conditions for this generalized spreading model to converge towards the infection-free equilibrium (i.e., to eradicate the spread) with a given exponential decay rate. We have specifically provided a \emph{transient} analysis of the \emph{stochastic} spreading dynamics over arbitrary networks without relying on mean-field approximations. Our results illustrate that the particular shape of the transmission/recovery distribution heavily influences the exponential decay rate of the convergence towards the infection-free equilibrium. {Through numerical simulations, we have specifically observed that our results allow to theoretically confirm some observations previously obtained by numerical simulations and mean-field approximations.}

{A possible direction for future research is considering time-varying (temporal) networks~\cite{Masuda2016b,Holme2015b} in which transmission and recovery events follow non-Poissonian distributions. Another interesting research direction is developing an optimal resource allocation strategy for non-Markovian epidemic spreading processes. Although we can find in the literature various research efforts~\cite{Wan2008IET,Preciado2014,Han2015a,Ogura2015a,AbadTorres2016} for designing containment methodologies for networked epidemic spreading processes, many of them are based on decay rates that are derived under the Markovian assumption on transmission and recovery events. In this direction, it is of practical interest to investigate how the non-Markovianity of spreading dynamics alters the optimal allocation strategies that have been investigated in the literature.}

\begin{IEEEbiography}
{Masaki Ogura} received the B.Sc.~degree in Engineering and M.Sc.~degree in Informatics from Kyoto University, and the Ph.D. degree in Mathematics from Texas Tech University. He was a Postdoctoral Researcher with the Department of Electrical and Systems Engineering at the University of Pennsylvania. He is currently an Assistant Professor with the Division of Information Science, Nara Institute of Science and Technology, Japan. His research interests include network science, dynamical systems, and convex optimizations with applications to epidemic control, consensus formation, and product development processes.
\end{IEEEbiography}

\begin{IEEEbiography}
{Victor M.~Preciado} received the PhD degree in electrical engineering and computer science from the Massachusetts Institute of Technology, in 2008. He is currently an associate professor of electrical and systems engineering with the University of Pennsylvania, where he is a member of the Networked and Social Systems Engineering (NETS) program, the Warren Center for Network and Data Sciences, and the Applied Math and Computational Science (AMCS) program. He is a recipient of the 2017 National Science Foundation Faculty Early Career Development (CAREER) Award and the 2018 Outstanding Paper Award from the IEEE Control Systems Magazine. His main research interests lie at the intersection of big data and network science; in particular, in using innovative mathematical and computational approaches to capture the essence of complex, high-dimensional dynamical systems. Relevant applications of this line of research can be found in the context of socio-technical networks, brain dynamical networks, healthcare operations, biological systems, and critical technological infrastructure.
\end{IEEEbiography}


\begin{thebibliography}{10}
\providecommand{\url}[1]{#1}
\csname url@samestyle\endcsname
\providecommand{\newblock}{\relax}
\providecommand{\bibinfo}[2]{#2}
\providecommand{\BIBentrySTDinterwordspacing}{\spaceskip=0pt\relax}
\providecommand{\BIBentryALTinterwordstretchfactor}{4}
\providecommand{\BIBentryALTinterwordspacing}{\spaceskip=\fontdimen2\font plus
\BIBentryALTinterwordstretchfactor\fontdimen3\font minus
  \fontdimen4\font\relax}
\providecommand{\BIBforeignlanguage}[2]{{%
\expandafter\ifx\csname l@#1\endcsname\relax
\typeout{** WARNING: IEEEtranS.bst: No hyphenation pattern has been}%
\typeout{** loaded for the language `#1'. Using the pattern for}%
\typeout{** the default language instead.}%
\else
\language=\csname l@#1\endcsname
\fi
#2}}
\providecommand{\BIBdecl}{\relax}
\BIBdecl

\bibitem{AbadTorres2016}
{J.~{Abad Torres}, S.~Roy, and Y.~Wan, ``{Sparse resource allocation for linear
  network spread dynamics},'' \emph{IEEE Transactions on Automatic Control},
  vol.~62, no.~4, pp. 1714--1728, 2017.}

\bibitem{Asmussen1996}
S.~Asmussen, O.~Nerman, and M.~Olsson, ``{Fitting phase-type distributions via
  the EM algorithm},'' \emph{Scandinavian Journal of Statistics}, vol.~23,
  no.~4, pp. 419--441, 1996.

\bibitem{Bailey1975}
N.~T.~J. Bailey, \emph{{The Mathematical Theory of Infectious Diseases and its
  Applications}}.\hskip 1em plus 0.5em minus 0.4em\relax Griffin, 1975.

\bibitem{Blythe1988}
S.~P. Blythe and R.~M. Anderson, ``{Variable infectiousness in HFV transmission
  models},'' \emph{Mathematical Medicine and Biology}, vol.~5, no.~3, pp.
  181--200, 1988.

\bibitem{Cator2013a}
{E.~Cator, R.~van~de Bovenkamp, and P.~{Van Mieghem},
  ``{Susceptible-infected-susceptible epidemics on networks with general
  infection and cure times},'' \emph{Physical Review E}, vol.~87, no.~6, p.
  062816, 2013.}

\bibitem{Cinlar1975a}
E.~\c{C}inlar, \emph{{Introduction to Stochastic Processes}}.\hskip 1em plus
  0.5em minus 0.4em\relax Prentice-Hall, 1975.

\bibitem{Chowell2014}
G.~Chowell and H.~Nishiura, ``{Transmission dynamics and control of Ebola virus
  disease (EVD): a review},'' \emph{BMC Medicine}, vol.~12, no.~1, p. 196,
  2014.

\bibitem{Commault2003}
{C.~Commault and S.~Mocanu, ``{Phase-type distributions and representations:
  Some results and open problems for system theory},'' \emph{International
  Journal of Control}, vol.~76, no.~6, pp. 566--580, 2003. }

\bibitem{Cox1955}
D.~R. Cox, ``{A use of complex probabilities in the theory of stochastic
  processes},'' \emph{Mathematical Proceedings of the Cambridge Philosophical
  Society}, vol.~51, no.~02, pp. 313--319, 1955.

\bibitem{DarabiSahneh2013}
F.~{Darabi Sahneh}, C.~Scoglio, and P.~{Van Mieghem}, ``{Generalized epidemic
  mean-field model for spreading processes over multilayer complex networks},''
  \emph{IEEE/ACM Transactions on Networking}, vol.~21, no.~5, pp. 1609--1620,
  2013.

\bibitem{Doerr2013}
{C.~Doerr, N.~Blenn, and P.~{Van Mieghem}, ``{Lognormal infection times of
  online information spread.}'' \emph{PLoS ONE}, vol.~8, no.~5, p. e64349,
  2013.}

\bibitem{Ganesh2005}
{A.~Ganesh, L.~Massouli\'e, and D.~Towsley, ``{The effect of network topology on
  the spread of epidemics},'' in \emph{24th Annual Joint Conference of the IEEE
  Computer and Communications Societies}, 2005, pp. 1455--1466.}

\bibitem{Han2015a}
S.~Han, V.~M. Preciado, C.~Nowzari, and G.~J. Pappas, ``{Data-driven network
  resource allocation for controlling spreading processes},'' \emph{IEEE
  Transactions on Network Science and Engineering}, vol.~2, no.~4, pp.
  127--138, 2015.

\bibitem{Holme2015b}
{P.~Holme, ``{Modern temporal network theory: a colloquium},'' \emph{The
  European Physical Journal B}, vol.~88, no.~9, p. 234, 2015.}

\bibitem{Horn1990}
R.~Horn and C.~Johnson, \emph{{Matrix Analysis}}.\hskip 1em plus 0.5em minus
  0.4em\relax Cambridge University Press, 1990.

\bibitem{Jo2014}
H.-H. Jo, J.~I. Perotti, K.~Kaski, and J.~Kert{\'{e}}sz, ``{Analytically
  solvable model of spreading dynamics with non-Poissonian processes},''
  \emph{Physical Review X}, vol.~4, no.~1, p. 011041, 2014.

\bibitem{Keeling2002}
M.~J. Keeling and B.~T. Grenfell, ``{Understanding the persistence of measles:
  Reconciling theory, simulation and observation},'' \emph{Proceedings of the
  Royal Society B: Biological Sciences}, vol. 269, no. 1489, pp. 335--343,
  2002.

\bibitem{Khanafer2016}
{A.~Khanafer, T.~Ba\c{s}ar, and B.~Gharesifard, ``{Stability of epidemic models
  over directed graphs: A positive systems approach},'' \emph{Automatica},
  vol.~74, pp. 126--134, 2016.}

\bibitem{Kiss2015}
I.~Z. Kiss, G.~R\"ost, and Z.~Vizi, ``{Generalization of pairwise models to
  non-Markovian epidemics on networks},'' \emph{Physical Review Letters}, vol.
  115, no.~7, p. 078701, 2015.

\bibitem{Lajmanovich1976}
{A.~Lajmanovich and J.~A. Yorke, ``{A deterministic model for gonorrhea in a
  nonhomogeneous population},'' \emph{Mathematical Biosciences}, vol.~28, no.
  1976, pp. 221--236, 1976.}

\bibitem{Lerman2010}
K.~Lerman and R.~Ghosh, ``{Information contagion: An empirical study of the
  spread of news on Digg and Twitter social networks},'' in \emph{Fourth
  International AAAI Conference on Weblogs and Social Media}, 2010, pp. 90--97.

\bibitem{Limpert2001}
{E.~Limpert, W.~A. Stahel, and M.~Abbt, ``{Log-normal distributions across the
  sciences: keys and clues},'' \emph{BioScience}, vol.~51, no.~5, pp. 341--352,
  2001.}



\bibitem{Mai2018}
{V.~S. Mai, A.~Battou, and K.~Mills, ``{Distributed algorithm for suppressing
  epidemic spread in networks},'' \emph{IEEE Control Systems Letters}, vol.~2,
  no.~3, pp. 555--560, 2018. }

\bibitem{Masuda2016b}
N.~Masuda and R.~Lambiotte, \emph{{A Guide to Temporal Networks}}.\hskip 1em
  plus 0.5em minus 0.4em\relax World Scientific Publishing, 2016.

\bibitem{Min2013}
B.~Min, K.-I. Goh, and I.-M. Kim, ``{Suppression of epidemic outbreaks with
  heavy-tailed contact dynamics},'' \emph{Europhysics Letters}, vol. 103,
  no.~5, p. 50002, 2013.

\bibitem{NISHIURA2007}
{H.~Nishiura and M.~Eichner, ``{Infectiousness of smallpox relative to disease
  age: estimates based on transmission network and incubation period},''
  \emph{Epidemiology and Infection}, vol. 135, no.~07, pp. 1145--1150, 2007.}

\bibitem{Nowzari2015}
{C.~Nowzari, M.~Ogura, V.~M. Preciado, and G.~J. Pappas, ``{A general class of
  spreading processes with non-Markovian dynamics},'' in \emph{54th IEEE
  Conference on Decision and Control}, 2015, pp. 5073--5078.}

\bibitem{Nowzari2015a}
C.~Nowzari, V.~M. Preciado, and G.~J. Pappas, ``{Analysis and control of
  epidemics: A survey of spreading processes on complex networks},'' \emph{IEEE
  Control Systems}, vol.~36, no.~1, pp. 26--46, 2016.

\bibitem{Ogura2015c}
M.~Ogura and V.~M. Preciado, ``{Stability of spreading processes over
  time-varying large-scale networks},'' \emph{IEEE Transactions on Network
  Science and Engineering}, vol.~3, no.~1, pp. 44--57, 2016. 

\bibitem{Ogura2015i}
------, ``{Epidemic processes over adaptive state-dependent networks},''
  \emph{Physical Review E}, vol.~93, p. 062316, 2016. 

\bibitem{Ogura2016l}
------, ``{Optimal Containment of Epidemics in Temporal and Adaptive
  Networks},'' in \emph{Temporal Networks Epidemiology}.\hskip 1em plus 0.5em
  minus 0.4em\relax Springer, 2017, pp. 241--266. 

\bibitem{Ogura2015a}
{------, ``{Optimal design of switched networks of positive linear systems via
  geometric programming},'' \emph{IEEE Transactions on Control of Network
  Systems}, vol.~4, no.~2, pp. 213--222, 2017.}

\bibitem{Ogura2017}
------, ``{Second-order moment-closure for tighter epidemic thresholds},''
  \emph{Systems \& Control Letters}, vol. 113, pp. 59--64, 2018.

\bibitem{Oksendal2003}
B.~{\O}ksendal, \emph{{Stochastic Differential Equations}}.\hskip 1em plus
  0.5em minus 0.4em\relax Springer, 2003.

\bibitem{Pare2018}
{P.~E. Par\'e, C.~L. Beck, and A.~Nedic, ``{Epidemic processes over time-varying
  networks},'' \emph{IEEE Transactions on Control of Network Systems}, vol.~5,
  no.~3, pp. 1322--1334, 2018.}

\bibitem{Pare2018a}
{P.~E. Par\'e, J.~Liu, C.~L. Beck, B.~E. Kirwan, and T.~Ba\c{s}ar, ``{Analysis,
  estimation, and validation of discrete-time epidemic processes},'' \emph{IEEE
  Transactions on Control Systems Technology}, 2018.}

\bibitem{Pastor-Satorras2015a}
R.~Pastor-Satorras, C.~Castellano, P.~{Van Mieghem}, and A.~Vespignani,
  ``{Epidemic processes in complex networks},'' \emph{Reviews of Modern
  Physics}, vol.~87, no.~3, pp. 925--979, 2015.

\bibitem{Pellis2015}
L.~Pellis, T.~House, and M.~J. Keeling, ``{Exact and approximate moment
  closures for non-Markovian network epidemics},'' \emph{Journal of Theoretical
  Biology}, vol. 382, pp. 160--177, 2015.

\bibitem{Preciado2014}
V.~M. Preciado, M.~Zargham, C.~Enyioha, A.~Jadbabaie, and G.~J. Pappas,
  ``{Optimal resource allocation for network protection against spreading
  processes},'' \emph{IEEE Transactions on Control of Network Systems}, vol.~1,
  no.~1, pp. 99--108, 2014.

\bibitem{Roy2012}
S.~Roy, M.~Xue, and S.~K. Das, ``{Security and discoverability of spread
  dynamics in cyber-physical networks},'' \emph{IEEE Transactions on Parallel
  and Distributed Systems}, vol.~23, no.~9, pp. 1694--1707, 2012.

\bibitem{Ruhi2016a}
N.~A. Ruhi, C.~Thrampoulidis, and B.~Hassibi, ``{Improved bounds on the
  epidemic threshold of exact SIS models on complex networks},'' in \emph{55th
  IEEE Conference on Decision and Control}, 2016, pp. 3560--3565.

\bibitem{Santos2015}
{A.~Santos, J.~M.~F. Moura, and J.~M.~F. Xavier, ``{Bi-virus SIS epidemics over
  networks: Qualitative analysis},'' \emph{IEEE Transactions on Network Science
  and Engineering}, vol.~2, no.~1, pp. 17--29, 2015.}

\bibitem{Starnini2017}
M.~Starnini, J.~P. Gleeson, and M.~Bogu$\tilde{\text{n}}$\'a, ``{Equivalence
  between Non-Markovian and Markovian dynamics in epidemic spreading
  processes},'' \emph{Physical Review Letters}, vol. 118, 2017.

\bibitem{Suess2010}
T.~Suess, U.~Buchholz, S.~Dupke, R.~Grunow, M.~{An der Heiden}, A.~Heider,
  B.~Biere, B.~Schweiger, W.~Haas, and G.~Krause, ``{Shedding and transmission
  of novel influenza virus A/H1N1 infection in households-Germany, 2009},''
  \emph{American Journal of Epidemiology}, vol. 171, no.~11, pp. 1157--1164,
  2010.

\bibitem{Mieghem2011a}
{P.~{Van Mieghem}, N.~Blenn, and C.~Doerr, ``{Lognormal distribution in the digg
  online social network},'' \emph{The European Physical Journal B}, vol.~83,
  pp. 251--261, 2011.}

\bibitem{VanMieghem2009a}
P.~{Van Mieghem}, J.~Omic, and R.~Kooij, ``{Virus spread in networks},''
  \emph{IEEE/ACM Transactions on Networking}, vol.~17, no.~1, pp. 1--14, 2009.

\bibitem{VanMieghem2013}
P.~{Van Mieghem} and R.~van~de Bovenkamp, ``{Non-Markovian infection spread
  dramatically alters the susceptible-infected-susceptible epidemic threshold
  in networks},'' \emph{Physical Review Letters}, vol. 110, no.~10, p. 108701,
  2013.

\bibitem{Wan2008IET}
Y.~Wan, S.~Roy, and A.~Saberi, ``{Designing spatially heterogeneous strategies
  for control of virus spread},'' \emph{IET Systems Biology}, vol.~2, no.~4,
  pp. 184--201, 2008.

\end{thebibliography}
\end{document}